\documentclass{article}

\usepackage{amssymb}
\usepackage{eucal}
\usepackage{amsmath}
\usepackage{amsthm}

\numberwithin{equation}{section}
\newtheorem{thm}{\bf Theorem}[section]
\newtheorem{lem}[thm]{\bf Lemma}

\theoremstyle{remark}
\newtheorem{rem}{\bf Remark}[section]

\title{A note on stability of the vertical uniform rotations \\
of the heavy top}
\author{Dan Com\u anescu\\
{\small Department of Mathematics, West University of Timi\c soara}\\
{\small Bd. V. P\^ arvan, No 4, 300223 Timi\c soara, Rom\^ ania}\\
{\small E-mail addresses: comanescu@math.uvt.ro}}
\date{}

\begin{document}

\maketitle


\begin{abstract}
We prove that the stability problem of a vertical uniform rotation of a heavy top is completely solved by using the linearization method and the conserved quantities of the differential system which describe the rotation of the heavy top.
\end{abstract}



\section{Introduction}

A classical problem in mechanics is the study of the rotation of a heavy rigid body with a fixed point. An important special case is the case of symmetric top or Lagrangian top, see \cite{arnold}, for which the inertia ellipsoid at fixed point is an ellipsoid of revolution and whose center of gravity lies on the axis of symmetry. The rotation of the heavy top is governed by the differential system
\begin{equation}\label{top}
\left\{%
\begin{array}{ll}
\dot{\vec{M}}=\vec{M}\times \mathbb{I}^{-1}\vec{M} +mg\vec{\gamma}\times\vec{r}_G \\
\dot{\vec{\gamma}}=\vec{\gamma}\times\mathbb{I}^{-1}\vec{M},\end{array}%
\right.
\end{equation}
where $m$ is the mass of the symmetric top, $g$ is the gravitational acceleration, $\vec{r}_G$ is the vector with the initial point in the fixed point $O$ and the terminal point in the center of gravity $G$, $\mathbb{I}$ is the moment of inertia tensor at the point $O$, $\vec{M}$ is the angular momentum vector and $\vec{\gamma}$ is the unit vector of the direction of the gravitational field. Also, one can use the equivalent description with the state parameters $\vec{\omega}$ and $\vec{\gamma}$, where $\vec{\omega}$ is the angular velocity vector and $\vec{M}=\mathbb{I}\vec{\omega}$. We denote by $A=B$ and $C$ the principal moments of inertia. For the following considerations we use a body frame for which the axes are principal axes of inertia and $G$ has the components $(0,0,z)$ with $z>0$. The matrix of the moment of inertia tensor in this body frame has the form $\mathbb{I}=\hbox{diag}(A,A,C)$. In the above frame the angular momentum vector $\vec{M}$ has the components $M_1, M_2, M_3$ and the unit vector of the direction of the gravitational field $\vec{\gamma}$ has the components $\gamma_1,\gamma_2,\gamma_3$.
We have four conserved quantities:
$$H=\frac{1}{2}(\frac{M_1^2}{A}+\frac{M_2^2}{A}+\frac{M_3^2}{C})+mgz\gamma_3,\,\,\,C_1=\gamma_1^2+\gamma_2^2+\gamma_3^2,\,\,\,C_2=M_1\gamma_1+M_2\gamma_2+M_3\gamma_3,
\,\,\,\text{and}\,\,\,F=M_3.$$
It is easy to see that a vertical uniform rotation $(0,0,\mathfrak{M}_{3},0,0,1)$ of the top is an equilibrium point for the system \eqref{top}.

It is well known, see \cite{chetaev}, \cite{holm-marsden-ratiu-weinstein}, \cite{marsden-ratiu} and \cite{rouche}, that the condition $C^2\omega^2>4Amgz$ is a sufficient condition for stability of the equilibrium point $(0,0,\omega,0,0,1)$ when we use the state parameters $\vec{\omega}$ and $\vec{\gamma}$. This condition implies the following sufficient condition for the stability of the vertical uniform rotation $(0,0,\mathfrak{M}_{3},0,0,1)$,
\begin{equation}\label{condition}
\mathfrak{M}_{3}^2>4Amgz.
\end{equation}
The method used by N.G. Chetaev (see \cite{chetaev}) and presented in the paper \cite{rouche} construct a Lyapunov function of the form $\lambda_1 H+\lambda_2 C_1+\lambda_3 C_2+\lambda_4 F+\mu F^2$. In the papers \cite{holm-marsden-ratiu-weinstein} and \cite{marsden-ratiu} is used the energy-Casimir method which also  construct a Lyapunov function by using the conserved quantities $H,C_1,C_2$ and $F$.

In this paper we study the possibility to construct a Lyapunov function using the conserved quantities $H,C_1,C_2$ and $F$. We apply an algebraic method also used in the papers \cite{comanescu} and \cite{comanescu-1}. We prove that it is possible to construct in a neighborhood of the vertical uniform rotation $(0,0,\mathfrak{M}_{3},0,0,1)$ a Lyapunov function using the conserved quantities $H,C_1,C_2,F$ if and only if we have $\mathfrak{M}_{3}^2\geq 4Amgz$.
We recover the sufficient condition \eqref{condition} for the Lyapunov stability of the vertical uniform rotation. We prove that the condition
\begin{equation}\label{condition-equality}
\mathfrak{M}_{3}^2=4Amgz
\end{equation}
is also a sufficient condition for the Lyapunov stability.

In the papers \cite{holm-marsden-ratiu-weinstein} and \cite{marsden-ratiu} is noted that the condition $\mathfrak{M}_{3}^2<4Amgz$ implies the instability of the uniform rotation $(0,0,\mathfrak{M}_{3},0,0,1)$; more precisely the uniform rotation is not spectrally stable (the linearization has an eigenvalue with strictly positive real part).

The stability problem of a vertical uniform rotation of a heavy top is completely solved by using the conserved quantities $H,C_1,C_2,F$ and the linearization method.

\section{Stability of the vertical uniform rotations}

The stability of an equilibrium point with respect to a set of conserved quantities is a sufficient condition for Lyapunov stability. If an equilibrium point is not stable with respect to a set of conserved quantities, then we cannot construct a Lyapunov function using this set of conserved quantities. We remind some theoretical considerations, from the paper \cite{comanescu}. We consider an open set $D\subset\mathbb{R}^n$ and the locally Lipschitz function $f:D\rightarrow \mathbb{R}^n$ which generates the differential equation
\begin{equation}\label{ecuatie-diferentiala}
    \dot{x}=f(x)
\end{equation}
Let $x_e$ be an equilibrium point. A continuous function $V:D\rightarrow \mathbb{R}$ which satisfies $V(x_e)=0$ and $V(x)>0$ for every $x$ in a neighborhood of $x_e$ and $x\neq x_e$ is called a positive definite function at the equilibrium point $x_e$.
{\it The equilibrium point $x_e$ of \eqref{ecuatie-diferentiala} is stable with respect to the set of conserved quantities $\{F_1,...,F_k\}$ if there exists a continuous function $\Phi:\mathbb{R}^k\rightarrow \mathbb{R}$ such that $x\rightarrow \Phi(F_1,....,F_k)(x)-\Phi(F_1,....,F_k)(x_e)$ is a positive definite function at $x_e$.}
In the conditions of the above definition the function $x\rightarrow \Phi(F_1,....,F_k)(x)-\Phi(F_1,....,F_k)(x_e)$ is a Lyapunov function at the equilibrium point $x_e$ and we have the following results.
\begin{thm}\label{implication-stability}
If the equilibrium point $x_e$ of \eqref{ecuatie-diferentiala} is stable with respect to the set of conserved quantities $\{F_1,...,F_k\}$ then it is stable in the sense of Lyapunov.
\end{thm}

\begin{thm}\label{stability}
Let $x_e$ be an equilibrium point of \eqref{ecuatie-diferentiala} and $\{F_1,...,F_k\}$ a set of conserved quantities. The following statements are equivalent:
\begin{itemize}
\item[(i)] $x_e$ is stable with respect to the set of conserved quantities $\{F_1,...,F_k\}$;
\item [(ii)] $x\rightarrow ||(F_1,...,F_k)(x)-(F_1,...,F_k)(x_e)||$ is a positive definite function at $x_e$;
\item [(iii)] the system of equations $F_1(x)=F_1(x_e),...,F_k(x)=F_k(x_e)$ has no root besides $x_e$ in some neighborhood of $x_e$.
\end{itemize}
\end{thm}

Theorem \ref{stability} $(iii)$ offer an algebraic method to prove the Lyapunov stability of an equilibrium point. This method have been used in \cite{comanescu} and \cite{comanescu-1} for studying the stability problem of the uniform rotations of a torque-free gyrostat and also for studying the stability problem of the equilibrium states of a heavy gyrostat (Zhukovskii case).

In our case the algebraic system at the equilibrium point $(0,0,\mathfrak{M}_{3},0,0,1)$ is
\begin{equation}\label{algebraic}
    \frac{1}{2}(\frac{M_1^2}{A}+\frac{M_2^2}{A}+\frac{M_3^2}{C})+mgz\gamma_3=\frac{\mathfrak{M}_3^2}{2C}+mgz,\,\,\gamma_1^2+\gamma_2^2+\gamma_3^2=1,
    \,\,M_1\gamma_1+M_2\gamma_2+M_3\gamma_3=\mathfrak{M}_3,\,\,M_3=\mathfrak{M}_3.
\end{equation}
This system is equivalent with the following system
\begin{equation}\label{algebraic-1}
    M_1^2+M_2^2-2Amgz(1-\gamma_3)=0,\,\,\gamma_1^2+\gamma_2^2+\gamma_3^2=1,
    \,\,M_1\gamma_1+M_2\gamma_2-\mathfrak{M}_3(1-\gamma_3)=0.
\end{equation}
We introduce $u,\varphi,v,\theta$ which satisfies
$M_1=u\cos\varphi,\,\,M_2=u\sin\varphi,\,\,\gamma_1=v\cos\theta,\,\,\gamma_2=v\sin\theta$.
The algebraic system for $u,\varphi,v,\theta$ and $\gamma_3$ is
\begin{equation}\label{algebraic-2}
   u^2=2Amgz(1-\gamma_3),\,\,v^2=1-\gamma_3^2,
    \,\,uv\cos(\theta-\varphi)=\mathfrak{M}_3(1-\gamma_3).
\end{equation}

\begin{lem}\label{lema}
The solution $(0,0,\mathfrak{M}_{3},0,0,1)$ of the system \eqref{algebraic} is isolated in the set of the solutions if and only if
$\mathfrak{M}_{3}^2\geq 4Amgz$.
\end{lem}

\begin{proof}
Let $(M_1,M_2,\mathfrak{M}_3,\gamma_1,\gamma_2,\gamma_3)$ be a solution of \eqref{algebraic} in a ball with the center in $(0,0,\mathfrak{M}_{3},0,0,1)$ and a radius $R<1$, then $0<\gamma_3\leq 1$ (see \eqref{algebraic-2}). If $\gamma_3=1$ then we have that $u=v=0$ and consequently the solution is $(0,0,\mathfrak{M}_{3},0,0,1)$. If $0<\gamma_3<1$ then, by using \eqref{algebraic-2}, we have
$$\frac{|\mathfrak{M}_3|}{\sqrt{2Amgz}}=\sqrt{1+\gamma_3}\cdot |\cos(\varphi-\theta)|<\sqrt{2}.$$
We deduce that a necessary condition for a solution of \eqref{algebraic}, except $(0,0,\mathfrak{M}_{3},0,0,1)$, to be situated in the ball  with the center in $(0,0,\mathfrak{M}_{3},0,0,1)$ and a radius $R<1$ is $\mathfrak{M}_{3}^2< 4Amgz$. Consequently, if $\mathfrak{M}_{3}^2\geq 4Amgz$, then $(0,0,\mathfrak{M}_{3},0,0,1)$ is isolated in the set of the solutions of \eqref{algebraic}.

We suppose that $\mathfrak{M}_{3}^2< 4Amgz$ and consider a sequence $(\gamma_{3})_n$ which satisfy the conditions: $0<(\gamma_{3})_n<1$ and $(\gamma_{3})_n\rightarrow_{n\rightarrow\infty}1$. There exists the sequences $(\varphi_n)$ and $(\theta_n)$ such that $\sqrt{1+(\gamma_{3})_n}\cos(\varphi_n-\theta_n)=\frac{\mathfrak{M}_3}{\sqrt{2Amgz}}$. We obtain a sequence $(u_n,\varphi_n,v_n,\theta_n,(\gamma_{3})_n)$ of solutions of \eqref{algebraic-2} with $0<u_n \rightarrow_{n\rightarrow\infty}0$ and $0<v_n \rightarrow_{n\rightarrow\infty}0$.

Consequently, we obtain a nonconstant sequence $((M_{1})_n,(M_{2})_n,\mathfrak{M}_{3},(\gamma_{1})_n,(\gamma_{2})_n,(\gamma_{3})_n)$ of solutions of \eqref{algebraic} which converge to $(0,0,\mathfrak{M}_{3},0,0,1)$. We deduce that the solution $(0,0,\mathfrak{M}_{3},0,0,1)$ is not isolated in the set of the solutions of \eqref{algebraic}.
\end{proof}

Using Lemma \ref{lema},and Theorem \ref{implication-stability}, and Theorem \ref{stability} and linearization method (see \cite{holm-marsden-ratiu-weinstein} and \cite{marsden-ratiu}) we obtain the following results.

\begin{thm} Let $(0,0,\mathfrak{M}_{3},0,0,1)$ be a vertical uniform rotation of the system \eqref{top}.
\begin{itemize}
\item [(i)] The vertical uniform rotation is stable with respect to the set of conserved quantities $\{H,C_1,C_2,F\}$ if and only if $\mathfrak{M}_{3}^2\geq 4Amgz$.
\item [(ii)] The inequality $\mathfrak{M}_{3}^2\geq 4Amgz$ is a necessary and sufficient condition for the Lyapunov stability of the vertical uniform rotation.

\end{itemize}
\end{thm}

\begin{rem}
If we use the angular velocity vector $\vec{\omega}$ and the unit vector of the direction of the gravitational field $\vec{\gamma}$ to describe the rotation of the top, then the necessary and sufficient condition for the Lyapunov stability of the vertical uniform rotation $(0,0,\omega,0,0,1)$ is $C^2\omega^2\geq 4Amgz$.
\end{rem}
\medskip

\noindent {\bf Acknowledgments.} This work was supported by a grant of the Romanian National Authority for
Scientific Research, CNCS UEFISCDI, project number PN-II-RU-TE-2011-3-0006.

\end{document}